\documentclass{svmult}

\usepackage[dvips]{graphicx}
\usepackage{amssymb, amsmath}
\usepackage{makeidx}  
\usepackage{multicol} 
\usepackage[bottom]{footmisc}
\usepackage{color}

\makeindex   

\def\R{\mathbb{R}}

\begin{document}

\title*{Analytical and numerical results  for American style of perpetual put options through transformation into nonlinear stationary  Black-Scholes equations}
\titlerunning{Analytical and numerical results  for American style of perpetual put options}
\author{Maria do Ros\'ario Grossinho\inst{1}, Yaser Faghan \inst{1}\and 
Daniel \v{S}ev\v{c}ovi\v{c}\inst{2}}
\institute{Instituto Superior de Economia e Gest\~ao and CEMAPRE, Universidade de Lisboa, Portugal
\texttt{mrg@iseg.ulisboa.pt}\\
\and Dept.\ Applied Mathematics \& Statistics, Comenius University, 842 48  Bratislava, Slovakia. \texttt{sevcovic@fmph.uniba.sk}}

\maketitle

\begin{abstract}
We analyze and calculate the early exercise boundary for a class of stationary generalized Black-Scholes equations in which the volatility function depends on the second derivative of the option price itself. A motivation for studying the nonlinear Black Scholes equation with a nonlinear volatility arises from option pricing models including, e.g., non-zero transaction
costs, investors preferences, feedback and illiquid markets effects and risk from unprotected portfolio. We present a method how to transform the problem of American style of perpetual put options into a solution of an ordinary differential equation and implicit equation for the free boundary position. We finally present results of numerical approximation of the early exercise boundary, option price and their dependence on model parameters.
\end{abstract}



\section{Introduction}

In this paper, we are concerned with a financial option with no fixed maturity and no exercise limit, called the perpetual option. This type of an option, which can be exercised at any time, can be considered as the American style of an option. However, in this case, the time to maturity has no impact on the price of the option. From the mathematical point of view, this leads to a solution of the stationary  Black--Scholes problem.
More precisely,  the valuation problem is transformed into the free boundary problem that consists of the construction of the function  $V(S)$ together with the early exercise boundary point $\varrho$ satisfying the following conditions:
\begin{equation*}
\frac{1}{2} \sigma^2 S^2 \partial^2_S V + r S \partial_S V -r V =0, \qquad S > \varrho,
\end{equation*}
and
\begin{equation*}
V(\varrho) = E-\varrho, \quad 
\partial_S  V (\varrho) = -1,
\quad  V(+\infty) = 0
\end{equation*}
(c.f. \cite{journal6},\cite{monograph0}, \cite{monograph2}). The function $V$ is  defined in the domain $S > \varrho$, where $\varrho$ is the free boundary position. If the diffusion coefficient $\sigma>0$ is constant then we are, in fact, considering stationary solutions of the  classical linear Black--Scholes parabolic equation. However, we suppose that $\sigma$ depends on the asset price $S$ and the product of the asset price $S$ and  the second derivative (Gamma) of the option price $H=S\partial_S^2V$, i.e. 
\begin{equation} 
\label{Y16}
\sigma=\sigma(S,H)=\sigma(S, S\partial _S^2V).
\end{equation} 
Let us mention our motivation for studying a nonlinear volatility of the form (\ref{Y16}). 
As it is known, the classical linear Black Scholes model (c.f \cite{journal14},\cite{monograph1}) was derived under several restrictive  assumptions that did not reflect the real market. In fact, no transaction costs were considered, the volatility was supposed to be constant, only liquid and complete  markets were considered.  Since then, several results have appeared in the literature relaxing these assumptions in order to overcome some drawbacks they created in practice. Regarding the volatility, it has been justified in practice that it is not constant and it may depend on the asset price itself. With this volatility function (\ref{Y16}), the classical model is generalized in such a way that it allows to consider non-zero transaction costs, market feedback and illiquid market effects due to large trading volumes, risk from investors preferences, etc. Mathematically, the problem will lose its linear feature, since the equation becomes a nonlinear partial differential equation (see e.g. \cite{monograph2}). 

One of the first nonlinear models  taking into account non-trivial transaction costs was proposed by Leland  \cite{journal16} for put or call options, later extended for more general types of option by Hoggard, Whalley and Wilmott \cite{journal17}. Avellaneda and Paras \cite{journal1} proposed the jumping volatility model in which the volatility changes with respect to the sign of the Gamma of the option. Frey and Patie \cite{journal9}, Frey and Stremme \cite{journal10}  developed models dealing with feedback and illiquid market impact due to large trading (see also \cite{journal17}). We also mention the so-called the Risk adjusted pricing model (RAPM) derived by Kratka \cite{journal15} and Janda\v{c}ka and \v{S}ev\v{c}ovi\v{c} \cite{journal14} in which both the transaction costs as well as the risk from unprotected portfolio are taken into account. In the RAPM model the volatility function depends on $H=S\partial^2_S V$ only, and it has the form:
\begin{equation}
 \sigma(H)^2 = \sigma^2_0 (1 + \lambda H^\frac13 ) = \sigma^2_0 (1 + \lambda (S\partial^2_S V)^\frac13 ),
 \label{intro-RAPM}
\end{equation}
where $\sigma_0>0$ is the constant historical volatility of the underlying asset and $\lambda$ is a model parameter depending on the transaction cost rate and the unprotected portfolio risk exposure. Recently, explicit solutions to European style of options described by the nonlinear Black--Scholes equation with varying volatility have been derived by Bordag  \emph{et al.} \cite{journal5} for the Frey and Patie as well as the RAPM models. 

Barnes and Soner \cite{journal3} proposed a model assuming that investor's preferences are shown by an exponential utility function. In this model, the volatility function depends on $H=S\partial^2_S V$ as well as $S$, and it has the following form:
\begin{equation}
\sigma(S, H)^2 
= \sigma^2_0 \left(1+\Psi(a^2 S H )\right)
= \sigma^2_0 \left(1+\Psi(a^2 S^2\partial^2_S V)\right),
\label{intro-barles}
\end{equation}
where the function $\Psi$ is the unique solution to the ODE: $\Psi^\prime(x) = (\Psi(x)+1)/(2 \sqrt{x\Psi(x)} -x), \Psi(0)=0$ and $a\ge 0$ is a constant depending transaction costs and investor's risk aversion parameter (see \cite{journal3} for details). Notice that $\Psi(x)\ge 0$ for all $x\ge 0$ and it has  the following asymptotic: $\Psi(x)=O(x^\frac{1}{3})$ for  $x\to 0$ and $\Psi(x)=O(x)$ for $x\to\infty$. 

Finally, we also mention the nonlinear volatility model developed by Mariani and Rial, Amster, Averbuj  \cite{journal0}, where transaction costs depend on the volume of trading assets in a linear decreasing way. Recently, it was generalized for arbitrary transaction cost functions by \v{S}ev\v{c}ovi\v{c} and \v{Z}it\v{n}ansk\'a in the paper \cite{journal21}.

The paper is organized as follows. In the next section, we recall the mathematical formulation of the perpetual American put option pricing model. Furthermore, we prove the existence and uniqueness of a solution to the free boundary problem. We derive a formula for the option price and a single implicit equation for the free boundary position $\varrho$. In Section 3 we construct suitable sub-- and supper--solutions based on Merton's explicit solutions with constant volatility. Finally, in Section 4, we present computational results of the free boundary position $\varrho$, the option price $V(S)$ and their dependence on model parameters.

\section{Perpetual American Put Option}
In this section we analyze the problem of the American style of perpetual put options. As referred previously, perpetual options are financial options with no fixed maturity and no exercise limit. As they can be exercised at any time, they have infinite maturity $T=+\infty$.

Consider the American style of a put option with the volatility $\sigma$ of the form (\ref{Y16}). Suppose that there exist a limit of the solution $V$ and an early exercise boundary position $S_f$ for the maturity $T\to\infty$. 
The pair consisting of the limiting price $V = V(S)=\lim_{T-t\to\infty}V(S,t)$ and the limiting early exercise boundary position $\varrho=\lim_{T-t\to\infty}S_f(t)$ of the  perpetual put option is a solution to the stationary nonlinear Black--Scholes problem (c.f. \cite{journal11}):
\begin{equation}
\frac{1}{2} \sigma(S, S\partial^2_S V) S^2 \partial^2_S V + r S \partial_S V -r V =0, \qquad S > \varrho,
\label{pperpetual}
\end{equation}
and
\begin{equation}
 V(\varrho) = E-\varrho, \quad 
\partial_S  V (\varrho) = -1,
\quad  V(+\infty) = 0
\label{pperpetual-bc}
\end{equation}
(c.f. \cite{monograph1}, \cite{monograph2}, \cite{journal19}).
We shall prove that under certain assumptions made on the volatility function the perpetual American put option problem (\ref{pperpetual})--(\ref{pperpetual-bc}) has the unique solution $(V(.),\varrho)$. We will present its explicit formula for the case when $\sigma=\sigma(H)$, i.e. the volatility depends on the term $H=S\partial^2_S V$ only. Furthermore, we will also present comparison results with explicit Merton's solutions recently obtained by the authors in  \cite{journal23}.

Throughout the paper we will assume that the volatility function $\sigma=\sigma(S,H)$ fulfills the following assumption:

\medskip
\noindent\textbf{Assumption 1.} {\it The volatility function $\sigma=\sigma(S,H)$ in (\ref{pperpetual}) is assumed to be a $C^1$ smooth nondecreasing function in the $H>0$ variable and $\sigma(S,H)\ge \sigma_0>0$ for any $S>0$ and $H\ge 0$ where $\sigma_0$ is a positive constant.}

\medskip
If we extend the volatility function $\sigma(S,H)$ by $\sigma(S,0)$ for negative values of $H$, i.e. $\sigma(S,H)=\sigma(S,0)$ for $H\le 0$ then the function
\begin{equation*}
\R\ni H\longmapsto \frac{1}{2}\sigma(S, H)^2 H \in \R
\label{nondecreasing}
\end{equation*}
is strictly increasing and therefore there exists the unique inverse function $\beta:\R\to \R$ such that 
\begin{equation}\label{Y0}
\frac12 \sigma(S, H)^2 H  = w \qquad\hbox{if and only if}\qquad H= \beta(x, w), \quad \hbox{where}\ \ S=e^x.
\end{equation}
Notice that the function $\beta$ is  a continuous and increasing function such that $\beta(0) =0$. 

\medskip
Concerning the inverse function we have the following useful lemma:

\begin{lemma}\label{lemma-beta}
Assume the volatility function $\sigma(S,H)$ satisfies Assumption 1. Then the inverse function $\beta$ has the following properties:
\begin{enumerate}
 \item $\beta(x,0)=0$ and $\frac{\beta(x,w)}{w}\le \frac{2}{\sigma^2_0}$ for all $x,w\in \R$;
 \item $\beta^\prime_w(x,w)\le  \frac{2}{\sigma^2_0}$ for all $x\in \R$ and $w>0$.
\end{enumerate}
\end{lemma}

\begin{proof}
Clearly, $\beta(x,0)=0$. For $w>0$ we have $\beta(x,w)>0$ and $w=\frac12 \sigma(e^x, \beta(x,w))^2 \beta(x,w)\ge \frac{\sigma^2_0}{2} \beta(x,w)$ and so  $\frac{\beta(x,w)}{w}\le \frac{2}{\sigma^2_0}$. If $w<0$ then $\beta(x,w)<0$ and we can proceed similarly as before. 

Differentiating the equality $w=\frac12 \sigma(e^x, \beta(x,w))^2 \beta(x,w)\ge \frac{\sigma^2_0}{2} \beta(x,w)$ with respect to $w>0$ yields:
\[
 1=\frac12 \sigma^2(e^x, \beta(x,w)) \beta^\prime_w(x,w)
 + \partial_H \left(\frac12 \sigma(e^x, H)^2  \right) H \ge \frac12 \sigma^2_0 \beta^\prime_w(x,w)
\]
for $H=\beta(x,w)>0$ and the proof of the second statement of Lemma follows.

\end{proof}

\medskip
The key step how to solve the perpetual American put option problem (\ref{pperpetual})--(\ref{pperpetual-bc}) consists in introduction of the following variable:
\begin{equation}
 W(x) = \frac{r}{S} \left( V(S) - S\partial_S V(S) \right)\quad \hbox{where}\ \ S=e^x.
\end{equation}

\begin{lemma}\label{lemma-W}
Let $x^0\in\R$ be given. The function $V(S)$ is a solution to equation (\ref{pperpetual}) for $S>\varrho=e^{x_0}$ satisfying the boundary condition:
\[
 V(S) - S \partial_S V(S) =E, \quad \hbox{at}\ \ S=\varrho,
\]
iff and only if the transformed function $W(x)$ is a solution to the initial value problem for the ODE:
\begin{eqnarray}
 && \partial_x W(x) = -W(x) -r\beta(x, W(x)), \quad x>x_0, \label{Eq-W}
\\
&&W(x_0) = r E e^{-x_0}. \nonumber 
\end{eqnarray}

\end{lemma}

\begin{proof}
As $\partial_x = S\partial_S$ we obtain
\begin{eqnarray*}
 \partial_x W(x) &=&  r S\partial_S ( S^{-1} V(S) - \partial_S V(S))
= r S S^{-1} \partial_S V(S)  -r S^{-1} V(S) - r S\partial^2_S V(S)
 \\
&=& - W(x) -r S\partial^2_S V(S) = -W(x) - r\beta(x,W(x)),
\end{eqnarray*}
because $\beta(x,W(x))=H\equiv S\partial^2_S V(S)$ if and only if $\frac12 \sigma(S, H)^2 H  = W(x)$ and $V$ solves (\ref{pperpetual}), i.e. 
\[
\frac12 \sigma(S, H)^2 H + \frac{r}{S} \left( S\partial_S V(S) - V(S) \right)=0.
\]
Finally, $W(x_0) = \frac{r}{S} \left( V(S) - S\partial_S V(S) \right) = r E e^{-x_0}$ where 
$S=\varrho=e^{x_0}$, as claimed.
\end{proof}

\medskip
Notice the equivalence of conditions:
\begin{equation}
  V(S) - S \partial_S V(S) =E \ \hbox{and}\ V(S)=E-S\quad \Longleftrightarrow\quad 
 \partial_S V(S) =-1 \ \hbox{and}\ V(S)=E-S.
 \label{BC-equiv}
\end{equation}

\medskip
Concerning the solution $W$ of the ODE (\ref{Eq-W}) we have the following auxiliary result:

\begin{lemma}\label{lemma-W-prop}
Assume $x^0\in\R$. Let $W=W_{x_0}(x)$ be the unique solution to the ODE (\ref{Eq-W}) for $x\in\R$ satisfying the boundary condition $W(x_0) = r E e^{-x_0}$ at the initial point $x_0$. Then 
\begin{enumerate}
 \item $W_{x_0}(x)>0$ for any $x\in\R$,
 \item the function $x_0\mapsto W_{x_0}(x)$ is increasing in the $x_0$ variable for any $x\in\R$,
 \item if the volatility function depends on $H=S\partial^2_S V$ only, i.e. $\sigma=\sigma(H)$, then 
\[
W_{x_0}(x) =F^{-1} (x_0-x)\quad \hbox{where}\quad  
F(W) = \int_{W_0}^W \frac{1}{w+r\beta(w)}dw, \qquad W_0=W(x_0)=rE e^{-x_0}.
\]

\end{enumerate}

\end{lemma}

\begin{proof}
According to Lemma~\ref{lemma-beta} we have $\beta(x,w)/w\le 2/\sigma^2_0$ for any $x\in\R$ and $w\not=0$. Hence
\[
\partial_x |\ln(W(x)| = -\left( 1+ r\frac{\beta(x,W(x))}{W(x)}\right) \ge -(1+\gamma) 
\]
where $\gamma=2r/\sigma^2_0$. Therefore
\[
 |W(x)| \ge |W(x_0)| e^{-(1+\gamma)(x-x_0)}>0,
\]
and this is why the function $W(x)$ does not change the sign. As $W(x_0)=rEe^{-x_0}>0$ we have $W_{x_0}(x) >0$ as well. 

The solution $W_{x_0}(x)$ to the ODE (\ref{Eq-W}) can be expressed in the form
\[
 W_{x_0}(x)=W_{x_0}(x_0) - \int_{x_0}^x (W_{x_0}(\xi) + r\beta(\xi, W_{x_0}(\xi)))d\xi.
= r E e^{-x_0} - \int_{x_0}^x \left(W_{x_0}(\xi) + r\beta(\xi, W_{x_0}(\xi))\right)d\xi.
\]
Let us introduce the auxiliary function 
\[
y(x) = \partial_{x_0} W_{x_0}(x).
\]
Then 
\begin{eqnarray*}
y(x) &=& -r E e^{-x_0} + W_{x_0}(x_0) + r\beta(x_0, W_{x_0}(x_0))
 -  \int_{x_0}^x \left(1 + r\beta^\prime_w (\xi, W_{x_0}(\xi))\right) y(\xi)d\xi
\\
&=& r\beta(x_0, W_{x_0}(x_0)) - \int_{x_0}^x \left(1 + r\beta^\prime_w (\xi, W_{x_0}(\xi)) \right)y(\xi)d\xi.
\end{eqnarray*}
Hence $y$ is a solution to the ODE:
\begin{eqnarray}
&& \partial_x y(x) = 
- \left(1 + r\beta^\prime_w (x, W_{x_0}(x)) \right)y(x), \quad x\in\R, \label{Eq-y}
\\
&&y(x_0) = r \beta(x_0, r E e^{-x_0} )>0. \nonumber 
\end{eqnarray}
With regard to Lemma~\ref{lemma-beta} we have $\beta^\prime_w (x, W_{x_0}(x))\le 2/\sigma^2_0$. Therefore the function $y$ is a solution to the differential inequality:
\[
\partial_x y(x) \ge 
- (1 + \gamma) y(x), \quad x\in\R, 
\]
where $\gamma=2r/\sigma^2_0$. As a consequence we obtain
\begin{equation}
 |y(x)|\ge |y(x_0)| e^{-(1+\gamma)(x-x_0)} >0
 \label{y-ineq}
\end{equation}
and this is why the function $y(x)$ does not change the sign. Therefore $\partial_{x_0} W_{x_0}(x) = y(x) >0$ and the proof of the statement 2) follows. 

Finally, if $\sigma=\sigma(H)$ we have $\beta=\beta(w)$ and so
\[
 \partial_x F(W(x)) = \frac{1}{W(x)+r\beta(W(x))} \partial_x W(x) = -1.
\]
Hence $F(W(x))= F(W(x_0)) - (x-x_0) = x_0-x$ and the statement 3) follows. 

\end{proof}

\begin{lemma}\label{lemma-uniqueroot}
Under Assumption 1, there exists the unique root $x_0\in\R$ of the implicit equation
\begin{equation}
\int_{x_0}^\infty \beta(x, W_{x_0}(x))dx =1.
\label{x0-root}
\end{equation}

\end{lemma}

\begin{proof}
Denote $\phi(x_0) =  \int_{x_0}^\infty \beta(x, W_{x_0}(x))dx$. Then $\phi(\infty) =0$ and
\[
\phi^\prime(x_0) = -   \beta(x_0, W_{x_0}(x_0)) + \int_{x_0}^\infty \beta^\prime_w(x, W_{x_0}(x)) y(x)dx
\]
where $y(x)= \partial_{x_0} W_{x_0}(x)$ is the solution to (\ref{Eq-y}). That is 
$\partial_x y(x) =  - \left(1 + r\beta^\prime_w (x, W_{x_0}(x)) \right)y(x)$ and 
$y(x_0) =  r \beta(x_0, W_{x_0}(x_0) ) = r \beta(x_0, r E e^{-x_0} )$. Therefore
\[
\phi^\prime(x_0) = -   \beta(x_0, W_{x_0}(x_0)) 
-\frac{1}{r}\int_{x_0}^\infty \partial_x y(x) + y(x)dx
= -\frac{1}{r} y(\infty) -\frac{1}{r}\int_{x_0}^\infty y(x)dx
\le  -\frac{1}{r}\int_{x_0}^\infty y(x)dx.
\]
As $y(x)= \partial_{x_0} W_{x_0}(x) \ge y(x_0) e^{-(1+\gamma)(x-x_0)}$ we have
\[
\phi^\prime(x_0) \le -\frac{1}{r} \frac{y(x_0)}{1+\gamma} 
= - \frac{\beta(x_0, W_{x_0}(x_0))}{1+\gamma}.
\]
It means that the function $\phi$ is strictly decreasing. Since
\[
\frac{1}{2}\sigma(e^{x_0}, \beta(x_0, W_{x_0}(x_0)))^2 \beta(x_0, W_{x_0}(x_0)) =  W_{x_0}(x_0)
= r E e^{-x_0} \to +\infty \quad \hbox{as}\ x_0\to-\infty,
\]
we have $\lim_{x_0\to-\infty} \beta(x_0, W_{x_0}(x_0)) = \infty$ and therefore $\lim_{x_0\to-\infty} \phi^\prime(x_0) = -\infty$. Therefore $\phi(-\infty)=\infty$. In summary, there exists the unique root $x_0$ of the equation $\phi(x_0)=1$, as claimed.

\end{proof}

Now we are in a position to state our main result on unique solvability of
the perpetual American put option problem (\ref{pperpetual})--(\ref{pperpetual-bc}).

\begin{theorem}\label{theorem-main}
Assume the volatility function $\sigma$ satisfies Assumption 1. Then there exists the unique solution $(V(.), \varrho)$ to the perpetual American put option problem (\ref{pperpetual})--(\ref{pperpetual-bc}). The function $V(S)$ is given by
\[
 V(S) = \frac{S}{r}\int_{\ln S}^\infty W_{x_0}(x) dx, \quad\hbox{for}\ S\ge \varrho=e^{x_0},
\]
where $W_{x_0}(x)$ is the solution to the ODE (\ref{Eq-W}) and $x_0$ is the unique root of the implicit equation (\ref{x0-root}).
\end{theorem}

\begin{proof}
Differentiating the above expression for $V(S)$ we obtain
\begin{eqnarray*}
 \partial_S V(S) &=&  \frac{1}{r}\int_{\ln S}^\infty W_{x_0}(x) dx
 -\frac{1}{r} W_{x_0}(\ln S)
\\
 S\partial^2_S V(S) &=&  - \frac{1}{r}\left( W_{x_0}(x) + \partial_x W_{x_0}(x)\right) = \beta(x,  W_{x_0}(x)),
\end{eqnarray*}
where $x=\ln S$. Hence
\[
\frac{1}{2}\sigma(S, S\partial^2_S V)^2 S^2 \partial^2_S V +r S\partial_S V -r V 
= S \left(\frac{1}{2}\sigma(e^x, \beta(x,  W_{x_0}(x)))^2 \beta(x,  W_{x_0}(x)) -  W_{x_0}(x) \right) =0,
\]
i.e. $V(S)$ is the solution to (\ref{pperpetual}) for $S>\varrho= e^{x_0}$. 

Furthermore,
\[
\left[V(S) - S\partial_S V(S)\right]_{S=\varrho} = 
V(\varrho) -   \frac{\varrho}{r}\int_{\ln \varrho}^\infty W_{x_0}(x) dx
 + \frac{\varrho}{r} W_{x_0}(\ln \varrho)
=E\varrho e^{-\ln\varrho}=E,
\]
and,
\begin{eqnarray*}
V(\varrho) &=& \frac{\varrho}{r}\int_{\ln \varrho}^\infty W_{x_0}(x) dx
=  \frac{\varrho}{r}\int_{\ln \varrho}^\infty -\partial_x W_{x_0}(x) - r\beta(x,W_{x_0}(x))dx
\\
&=&  \frac{\varrho}{r} W_{x_0}(\ln \varrho) - \varrho \int_{\ln \varrho}^\infty \beta(x,W_{x_0}(x))dx = E - \varrho
\end{eqnarray*}
because $x_0$ is the unique solution to (\ref{x0-root}). With regard to the equivalence (\ref{BC-equiv}) we have $\partial_SV(S) = -1$ at $S=\varrho$. In summary, $(V(.),\varrho)$ is the unique solution to the perpetual American put option problem (\ref{pperpetual})--(\ref{pperpetual-bc}).
\end{proof}

\medskip

\begin{remark}
In the case the volatility function depends on $H=S\partial^2_SV$ only, i.e. $\sigma=\sigma(H)$, then equation (\ref{x0-root}) can be simplified by introducing the change of variables $w=W_{x_0}(x)$.  Indeed, $\beta=\beta(w)$ and $dw=\partial_x W_{x_0}(x) dx = -(W_{x_0}(x) + r\beta (W_{x_0}(x))) dx$. Therefore
\[
\int_{x_0}^\infty \beta(W_{x_0}(x))dx = -\int_{W_{x_0}(x_0)}^0 \frac{\beta(w)}{w+r\beta(w)} dw = \int_0^{\frac{rE}{\varrho}} \frac{\beta(w)}{w+r\beta(w)} dw.
\]
Equation   (\ref{x0-root}) can be rewritten in the following form
\begin{equation}
\int_0^{\frac{rE}{\varrho}} \frac{\beta(w)}{w+r\beta(w)} dw =1.
\label{x0-root-simplified}
\end{equation}
This is the condition for the free boundary position $\varrho$ recently derived by the authors in \cite{journal23}.

\end{remark}

\section{The Merton explicit solution, sub and super solutions}

In this section we recall recent results due to the authors \cite{journal23} dealing with comparison of the solution $(V(.),\varrho)$ to the perpetual American put option problem (\ref{pperpetual})--(\ref{pperpetual-bc}) for the case when the volatility function depends on $H=S\partial^2_SV$ only, i.e. $\sigma=\sigma(H)$. 

Suppose that the volatility $\sigma\equiv \sigma_0$ is constant, then for the function $V(S)$ and the limiting early exercise boundary position $\varrho$ the free boundary value problem  (\ref{pperpetual})--(\ref{pperpetual-bc}) has the explicit solution presented by Merton (c.f.  \cite{monograph1},\cite{monograph2}), which has the closed form:
\begin{equation}
V_\gamma(S) = \left\{
\begin{array}{ll}
 E-S,  &  0<S \le \varrho_\gamma,\\ 
 \frac{E}{1+\gamma} \left(\frac{S}{\varrho_\gamma}\right)^{-\gamma}, & S> \varrho_\gamma,
\end{array} 
\right.
\end{equation}
where 
\begin{equation}
\varrho_\gamma=E \frac{\gamma}{1+\gamma}, \qquad \gamma = \frac{2r}{\sigma_0^2}.
\label{gamma}
\end{equation}

Our next goal is to establish sub-- and super--solutions to the perpetual American put option pricing problem.  Let $\gamma>0$ is a positive constant and denote by $V_\gamma$ the explicit Merton solution defined before. It is clear that the pair $(V_\gamma(\cdot), \varrho_\gamma)$ is the explicit Merton solution with constant volatility $\sigma_0^2=2r/\gamma$.

Then, for the transformed function $W_\gamma(x)$ we have
\[
 W_\gamma(x) = r E \varrho_\gamma^\gamma e^{-(1+\gamma)x},\quad \hbox{for}\ 
 x=\ln S > x_{0\gamma} = \ln \varrho_\gamma.
\]
Furthermore $W_\gamma$ is a solution to the ODE:
\begin{equation}\label{Y9}
\partial_x W_\gamma + W_\gamma + \gamma W_\gamma =0.
\end{equation}
Applying the equation (\ref{Y9}) we can construct  a super-solution $W_{\gamma^+}$ and a sub-solution $W_{\gamma^-}$ to the solution $W$ of the equation:
\begin{equation*} 
\partial_x W = - W- r \beta(W)
\end{equation*}
using the Merton solution $W_\gamma$. Here $\gamma^+$ is the unique root of the equation
\begin{equation*}
\gamma^+ \sigma(1+\gamma^+)^2 = 2r  
\end{equation*}
and $\gamma^-$ satisfies
\begin{equation*}
 \gamma^- \sigma(0)^2 = 2r.
\end{equation*}

As a consequence, the following inequalities hold. For more details, we refer to \cite{journal23}.
\begin{equation}\label{Y11}
\left\{
\begin{array}{ll}
 \partial_x W_{\gamma^+}(x) \ge - W_{\gamma^+}(x) - r \beta(W_{\gamma^+}(x)),  &  \hbox{for}\ x> x_{0\gamma^+}=\ln\varrho_{\gamma^+},\\ 
\\
\partial_x W_{\gamma^-}(x) \le - W_{\gamma^-}(x) - r \beta(W_{\gamma^-}(x)), & \hbox{for}\ x> x_{0\gamma^-}=\ln\varrho_{\gamma^-}.
\end{array} 
\right.
\end{equation}
Moreover, it can be proved that
\begin{equation*}
\varrho_{\gamma^+} \le \varrho\le \varrho_{\gamma^-}.
\end{equation*} 
Since, for initial conditions we have $W_{\gamma^\pm}(x_{0_{\gamma^\pm}}) =\frac{ rE}{\varrho_{\gamma^\pm}�}$ and $W(x_0) = \frac{rE}{\varrho}$ and so
\begin{equation*}
W_{\gamma^-}(x_{0_{\gamma^-}})\le W(x_0)\le W_{\gamma^+}(x_{0_{\gamma^+}}).
\end{equation*}
Using the comparison principle for solutions of ordinary differential inequalities in (\ref{Y11}) we conclude
\begin{equation*}
W_{\gamma^-}(x)\le W(x)\le W_{\gamma^+}(x).
\end{equation*}
Then taking into account the explicit solution of the function $V(S)$ from Theorem~\ref{theorem-main} we present the following result:
\begin{theorem}\cite[Theorem 3]{journal23}
\label{theo-comparison}
Let $(V(\cdot),\varrho)$ be the solution to the perpetual American pricing problem  (\ref{pperpetual})--(\ref{pperpetual-bc}). Then 
for any $S\ge0$ we have 
\[
 V_{\gamma^-}(S) \le V(S) \le V_{\gamma^+}(S) 
\]
and 
\[
 \varrho_{\gamma^+} \le \varrho \le \varrho_{\gamma^-}
\]
where $(V_{\gamma^\pm}(.), \varrho_{\gamma^\pm})$ are explicit Merton's solutions with constant volatilities.
\end{theorem}

\section{Numerical approximation scheme and results}

In the last section, our aim is to present an efficient numerical scheme for constructing a solution to the perpetual American put option problem (\ref{pperpetual})--(\ref{pperpetual-bc}) for the case when the volatility function has the form: $\sigma=\sigma(H)$ where $H=S\partial^2_S V$. The numerical results were obtained by the authors in \cite{journal23}.

Our scheme is based on transformation $H=\beta(w)$, i.e. $w=\frac12\sigma(H)^2 H$ and $dw=\frac12 \partial_H(\sigma(H)^2 H) dH$ by using this we can rewrite the equation (\ref{x0-root-simplified}) for the free boundary position $\varrho$ as follows:
\begin{equation}
\int_0^{\beta(r E/\varrho)} \frac{H  }{\frac12\sigma(H)^2 H +r H} \frac12 \partial_H(\sigma(H)^2 H) d H =1.
\label{theo-rho-mod}
\end{equation}
Similarly, the expression (see Theorem~\ref{theorem-main}) for the price of the option can be rewritten in terms of the $H$ variable as follows:
\begin{equation}
V(S) = \frac{S}{r} \int_0^{\beta(F^{-1}(\ln(\varrho/S)))} 
\frac{\frac12\sigma(H)^2 H  }{\frac12\sigma(H)^2 H +r H} \frac12 \partial_H(\sigma(H)^2 H) d H.
\label{theo-V-mod}
\end{equation}
When the inverse function $\beta(w)$ is not given by a closed form formula by applying this transformation we can avoid computational complexity. 

In what follows we recall numerical results of computation of the solution to the perpetual American put option problem (\ref{pperpetual})--(\ref{pperpetual-bc}) for the RAPM model with the nonlinear volatility function of the form:
\begin{equation}\label{Y20}
\sigma(H)^2 = \sigma_0^2 \left(1+\lambda H^{\frac{1}{3}}\right),
\end{equation}
We propose the results of numerical calculation for the Risk adjusted pricing methodology model (RAPM). We would like to show the position of the free boundary $\varrho$ and the value of the perpetual option $V$ evaluated at exercise price $S=E$. The option values are computed for various values of the model $\lambda\in[0,2]$ for the RAPM model. The rest of the model parameters were chosen as: $r=0.1, E=100$ and $\sigma_0=0.3$. In computations shown in Tab.~\ref{tab:1} we present results of the free boundary position and the perpetual American put option price $V(E)$ for the RAPM model.

\begin{table}
\centering
\caption{The perpetual put option free boundary position $\varrho$ and the option price $V(S)$ evaluated at $S=E$ for various values of the model parameter $\lambda\ge0$ for the RAPM model (Source \cite{journal23}).} 
\label{tab:1}

\begin{tabular}{cccccccc}
\hline\noalign{\smallskip}
$\lambda$& 
0.00& 
0.20& 
0.40& 
0.60& 
1.20& 
1.60& 
2.00 \\
\noalign{\smallskip}\hline\noalign{\smallskip}
$\varrho$& 
68.9655& 
64.7181& 
61.2252& 
58.2647& 
51.1474&
47.2975& 
 44.5433\\
$V(E)$& 
13.5909& 
15.4853& 
 17.1580& 
18.6669 & 
22.5461& 
 24.7444  & 
 26.6804\\
\noalign{\smallskip}\hline
\end{tabular}


\end{table}

\begin{figure}
\centering
\includegraphics[height=5cm]{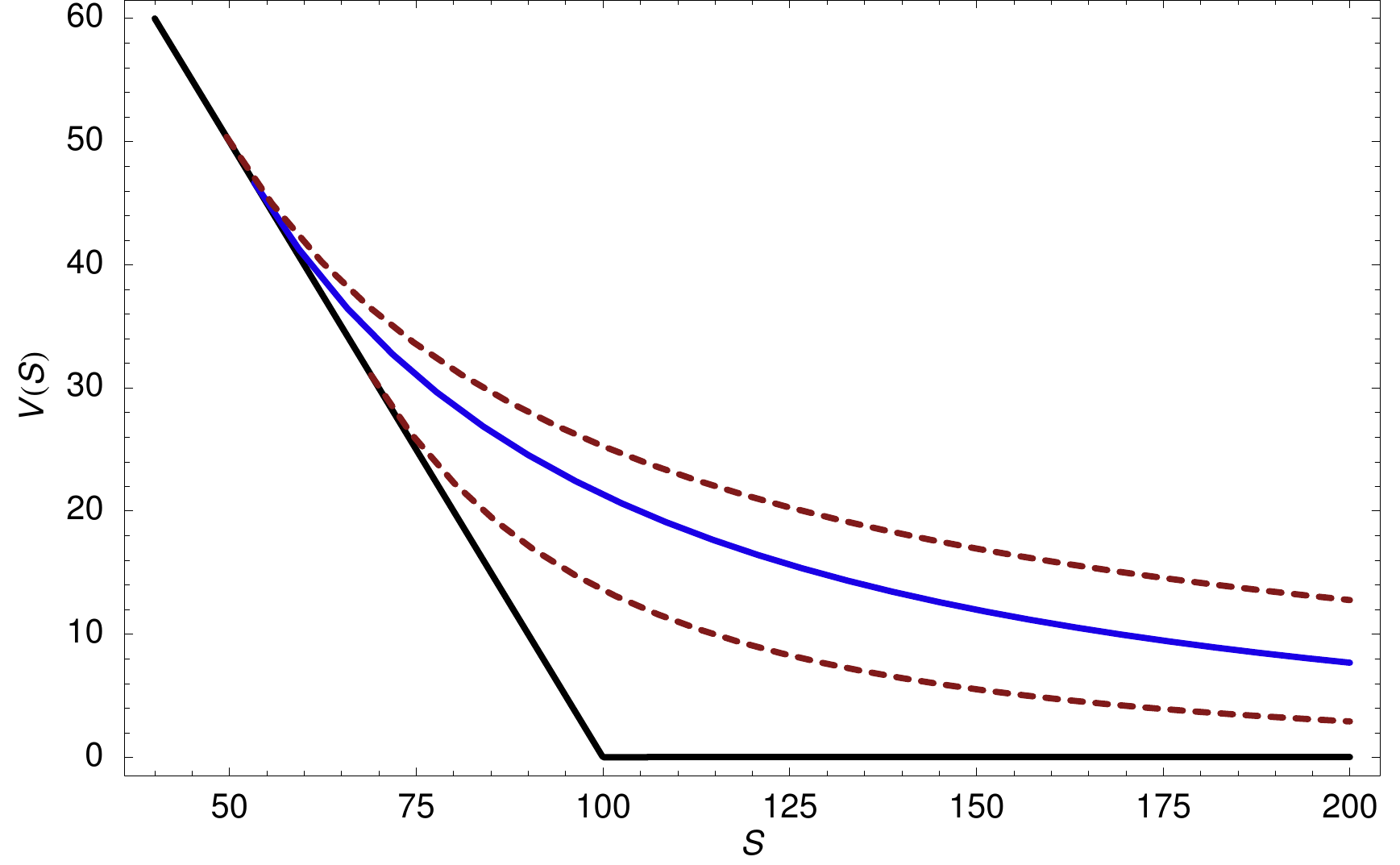}

\caption{Solid curve represents a graph of a perpetual American put option $V(S)$  for the RAPM model with $\lambda=1$. Sub- and super- solutions $V_{\gamma^-}$ and $V_{\gamma^+}$ are depicted by dashed curves. The model parameters: $r=0.1, E=100$ and $\sigma_0=0.3$ (Source \cite{journal23}).}
\label{obr-FreyRAPM-perpetual}
\end{figure}
Finally, in Fig.~\ref{obr-FreyRAPM-perpetual} we show the option price $V(S)$ for the Risk adjusted pricing methodology model with closed form explicit Merton's solutions with constant volatility.

\section{Conclusions}
In this paper we analyzed the problem of American style perpetual options when the nonlinear volatility is a function of the second derivative. We studied the free boundary problem that models this type of options, by transforming it into a single implicit equation for the free
boundary position and explicit integral expression for the option price. \bigskip


\noindent \textbf{Acknowledgements:} \medskip

\noindent This research was supported by the European Union in the FP7-PEOPLE-2012-ITN
project STRIKE - Novel Methods in Computational Finance (304617), the project CEMAPRE
MULTI/00491 financed by FCT/MEC through national funds and the Slovak research
Agency Project VEGA 1/0780/15.


%
%
%

%
%


\printindex
\end{document}